\documentclass[11pt]{article}
\usepackage[T2A]{fontenc}
\usepackage{ alphalph,etoolbox }
\usepackage{amsthm, amsmath, amssymb, amsbsy, amscd, amsfonts, latexsym, euscript,epsf, bm}
\newtheorem{thm}{Theorem}[section]
\usepackage{bbold, bm, stackrel}
\usepackage{epic,eepic}
\usepackage{texdraw}
\usepackage[dvips]{color}
\usepackage{color}
\usepackage{ dsfont }
\usepackage{graphicx}
\usepackage{exscale,relsize}
\usepackage{authblk}
\usepackage{url}
\usepackage{hyperref}
\usepackage{enumerate}
\usepackage{graphicx}
\newtheorem{proposition}[thm]{Proposition}

\newtheorem{theorem}[thm]{Theorem}

\theoremstyle{definition}
\newtheorem{definition}[thm]{Definition}
\newtheorem{remark}[thm]{Remark}

\DeclareMathOperator{\rank}{rank}

\DeclareMathOperator{\id}{id}

\newcommand{\dmp}{k}

\newcommand{\Mfd}{\mathbf{M}}
\newcommand{\xc}{\mathbf{x}}
\newcommand{\cl}{\colon}

\title{Yang--Baxter maps of KdV, NLS and DNLS type on division rings}
\date{}

\author{S. Konstantinou-Rizos\thanks{skonstantin84@gmail.com} ~and A.~A. Nikitina\thanks{anasta\_niki@mail.ru}}
\affil{Centre of Integrable Systems, P.G. Demidov Yaroslavl State University, Yaroslavl, Russia}

\patchcmd{\subequations}{\alph{equation}}{\alphalph{\value{equation}}}{}{}

\begin{document}

\maketitle

\begin{abstract}
We construct nocommutative set-theoretical solutions to the Yang--Baxter equation related to the KdV, the NLS and the derivative NLS equations. In particular, we construct several Yang--Baxter maps of KdV type and we show that one of them is completely integrable in the Liouville sense. Then, we construct a noncommutative KdV type Yang--Baxter map which can be squeezed down to the noncommutative discrete potential KdV equation. Moreover, we construct Darboux transformations for the noncommutative derivative NLS equation. Finally, we consider matrix refactorisation problems for noncommutative Darboux matrices associated with the NLS and the derivative NLS equation and we construct noncommutative maps. We prove that the latter are solutions to the Yang--Baxter equation.
\end{abstract}

\bigskip


\begin{quotation}
\end{quotation}

\begin{quotation}
\end{quotation}

\begin{quotation}
\end{quotation}

\begin{quotation}
\end{quotation}

\section{Introduction}\label{intro}
The Korteweg--de Vries (KdV) and the Nonlinear Schr\"odinger (NLS) equation are undoubtedly the most celebrated equations of Mathematical Physics. They are related to many aspects of integrability, namely they possess Lax pairs, Darboux and B\"acklund tranformations, they are solvable by the inverse scattering transformation and they admit soliton solutions. 

On the other hand, the Yang--Baxter equation has been very actively studied over the past few decades. It has applications in many diverse fields of Mathematics and Physics, ranging from statistical and quantum mechanics to topology and representation theory, and has earned its place in the list of most fundamental equations of Mathematical Physics. However, the most important application of the Yang--Baxter equation is its relation to the theory of Integrable Systems.

Yang--Baxter maps are set-theoretical solutions to the Yang--Baxter equation. Some applications of Yang--Baxter maps in the theory of integrable systems include their strict relation to  integrable systems of difference equations via the symmetries of the latter \cite{Pap-Tongas-Veselov}, their connection with integrable equations of mathematical physics \cite{Goncharenko-Veselov} and their association with nonlinear integrable PDEs via Darboux transformations \cite{Sokor-Sasha, GKM, MPW}. For particular applications of Yang--Baxter maps in relation with the KdV and the NLS equation, we indicatively refer to the derivation of the discrete potential KdV equation from a particular Yang--Baxter map \cite{Kouloukas}, the appearance of Yang--Baxter maps as soliton solutions of the matrix KdV equation \cite{Goncharenko-Veselov}, and the contruction of Yang--Baxter maps via Darboux transformations of the NLS and the derivative NLS equation \cite{Sokor-Sasha}.

The need to study the noncommutative solutions of the Yang--Baxter equation comes naturally from the continuously increasing popularity of noncommutative integrable systems, which have been in the centre of interest for many researchers in Mathematical Physics (see, e.g., \cite{Bobenko-Suris, Dimakis-Hoissen, Dimakis-Hoissen-2015, Doliwa-Noumi, Kupershmidt, Nijhoff-Capel, Nimmo, Talalaev}). Indeed, noncommutative solutions to the  Yang--Baxter equation have been  studied by leading scientists in the area of Integrable Systems (see, e.g., \cite{Doliwa-2014, Kassotakis-Kouloukas} and the references therein).

This paper aims to extend some of the results of papers \cite{Kouloukas}, \cite{Sokor-Kouloukas} and \cite{Sokor-Sasha}. Specifically, we shall adopt the approach of correspondences to construct Yang--Baxter maps from the same generator used to construct a KdV type of Yang--Baxter map in \cite{Kouloukas}, we shall construct a fully noncommutative analogue of the KdV lift Yang--Baxter map, which first appeared in \cite{Kouloukas} and was extended to an anticommutative (Grassmann) Yang--Baxter map. Moreover, we shall construct noncommutative analogues of the NLS and the derivative NLS Yang--Baxter maps which appeared in \cite{Sokor-Sasha} using noncommutative versions of Darboux transformations for NLS type equations.

In particular, what is new in this paper:
\begin{itemize}
    \item We demonstrate that the approach of correspondences \cite{Igonin-Sokor} gives rise to a wider class of Yang--Baxter maps, using a KdV type Yang--Baxter map as an illustrative example;
    \item We construct new KdV type Yang--Baxter maps, and we prove that one of them is completely integrable in the sense of Liouville;
    \item We construct a noncommutative version of the KdV type Yang--Baxter map \cite{Kouloukas} which can be squeezed down to the noncommutative discrete potential KdV equation;
    \item We construct a noncommutative version of the Darboux transformation for the derivative NLS equation \cite{SPS};
    \item We construct noncommutative versions of the NLS and derivative NLS type Yang--Baxter maps which appeared in \cite{Sokor-Sasha}.
\end{itemize}

The rest of the text is organised as follows. In the next section we give all the preliminary definitions in order for the text to be self-contained. In particular, we fix the notation we use throughout the text, we give the definitions of Yang--Baxter maps, parametric Yang--Baxter maps and their Lax representation. Additionally, we explain the relation between Yang--Baxter maps and matrix refactorisation problems and we give the definition of a completely integrable map in the Liouville sense. In Section \ref{KdV-type maps}, employing the approach  of correspondences, we construct  two novel Yang--Baxter maps of KdV type, and we show  that one of them is completely integrable in the Liouville sense. Furthermore, we  construct a noncommutative version of  a the ``KdV lift,'' which first appeared in \cite{Kouloukas}, and we  show that it satisfies the  Yang--Baxter equation.  In section \ref{NLS type  maps}, we first  construct a Darboux--B\"acklund transformation  related to the noncommutative derivative NLS equation. Then,  we employ  Darboux matrices of  NLS and  derivative NLS type  to construct noncommutative parametric Yang--Baxter maps. Finally, in section  \ref{conclusions}, we close with some concluding remarks  and present possible  directions for  further  research.

\section{Preliminaries}
\subsection{Notation}
Throughout the text:
\begin{itemize}
    \item By $\mathcal{X}$ we denote an arbitrary set, whereas by Latin italic letters (i.e. $x, y, u, v$ etc.) the elements of $\mathcal{X}$, with an exception of the `spectral parameter' which is denoted by the Greek letter $\lambda$. Moreover, by $\mathcal{X}^n$ we denote the Cartesian product of  $\mathcal{X}$ times itself $n$ times, i.e.
$\mathcal{X}^n=\underbrace{\mathcal{X}\times \mathcal{X}\times\dots\times \mathcal{X}}_{n}$.
    
    \item By $\mathfrak{R}$ we denote a noncommutative division ring, and its elements are denoted by bold italic Latin letters (i.e. $\bm{x}, \bm{y}, \bm{u}$ etc.). That is, $\mathfrak{R}$ is an associative algebra with multiplicative identity $1$ where commutativity with respect to mutliplication is not assumed ($\bm{x}\bm{y}\neq \bm{y}\bm{x}$), and every nonzero element $\bm{x}$ has an inverse $\bm{x}^{-1}$, i.e. $\bm{x}\bm{x}^{-1}=\bm{x}^{-1}\bm{x}=1$. The centre of a division ring will be denoted by $Z(\mathfrak{R})=\{a\in\mathfrak{R}:\forall\bm{x}\in\mathfrak{R},a\bm{x}=\bm{x}a\}$.
    
    \item Matrices will be denoted by capital Roman straight letters (i.e. ${\rm A}, {\rm B}, {\rm C}$) etc. Additionally, matrix operators are denoted by capital caligraphic letters (for instance, $\mathcal{L}=D_x-\rm{U}$).
\end{itemize}

\subsection{Set-theoretical Yang--Baxter equation}
Let $\mathcal{X}$ be a set. A map $Y:\mathcal{X}^2\rightarrow \mathcal{X}^2$ is called a \textit{Yang--Baxter map} if it satisfies the Yang--Baxter equation
\begin{equation}\label{eq_YB}
Y^{12}\circ Y^{13}\circ Y^{23}=Y^{23}\circ Y^{13}\circ Y^{12}.
\end{equation}
The terms $Y^{12}$, $Y^{13}$, $Y^{23}$ in \eqref{eq_YB} are maps $\mathcal{X}^3\to \mathcal{X}^3$ defined as follows 
$$
Y^{12}(x,y,z)=\big(u(x,y),v(x,y),z\big),~~~
Y^{23}(x,y,z)=\big(x,u(y,z),v(y,z)\big),~~~
Y^{13}(x,y,z)=\big(u(x,z),y,v(x,z)\big),
$$
where $x,y,z\in \mathcal{X}$.

Now, let $\mathcal{S}$ and $\mathcal{X}$ be sets. 
A \emph{parametric Yang--Baxter map} $Y_{a,b}$ is a family of maps
\begin{gather}
\label{ParamYB}
Y_{a,b}\colon \mathcal{X}\times \mathcal{X}\to \mathcal{X}\times \mathcal{X},\qquad
Y_{a,b}(x,y)=\big(u_{a,b}(x,y),\,v_{a,b}(x,y)\big),\quad x,y\in \mathcal{X},
\quad a,b\in\mathcal{S},
\end{gather}
depending on parameters $a,b\in\mathcal{S}$ 
and satisfying the parametric Yang--Baxter equation
\begin{gather}
\label{pybeq}
Y^{12}_{a,b}\circ Y^{13}_{a,c} \circ Y^{23}_{b,c}=
Y^{23}_{b,c}\circ Y^{13}_{a,c} \circ Y^{12}_{a,b}\quad\text{for all }\,a,b,c\in\mathcal{S}.
\end{gather}
Probably the most popular example of parametric Yang--Baxter map is the Adler map \cite{Adler}
$$
Y:(x,y)\rightarrow \left(y-\frac{a-b}{x+y}, x+\frac{a-b}{x+y}\right)
$$
which can be obtained from the discrete potential KdV equation
\begin{equation}\label{pkdv}
(f_{11}-f)(f_{01}-f_{10})=a-b.
\end{equation}
after a change of variables \cite{Pap-Tongas-Veselov}.

\subsection{Lax representation of Yang--Baxter maps}
Let $Y_{a,b}$ be a parametric Yang--Baxter map as defined in \eqref{ParamYB}. Suppose that $u=u_{a,b}(x,y)$ and $v=v_{a,b}(x,y)$ obey the equation
\begin{gather}
\label{eq-Lax}
{\rm L}_a(u){\rm L}_b(v)={\rm L}_b(y){\rm L}_a(x)
\end{gather}
for all values of $x,y,a,b,\lambda$. 
Then, following \cite{Veselov2}, we say that ${\rm L}_a(x)={\rm L}(x;a,\lambda)$ 
is a \emph{Lax matrix} for the parametric map $Y_{a,b}$.
Moreover, the matrix refactorisation problem \eqref{eq-Lax} is called a Lax representation for $Y_{a,b}$. 

However, not every map satisfying \eqref{eq-Lax} is a parametric Yang--Baxter map. In particular, we have the following.

\begin{proposition} (Kouloukas--Papageorgiou \cite{Kouloukas})\label{KP}
Let $Y_{a,b}$ be a map with Lax representation \eqref{eq-Lax}. If the following \textit{matrix trifactorisation problem}
\begin{gather}
\label{trifac}
{\rm L}_a(\hat{x}){\rm L}_b(\hat{y}){\rm L}_c(\hat{z})={\rm L}_a(x){\rm L}_b(y){\rm L}_c(z),\quad\text{for all}~~a,b,c\in\mathcal{S}
\end{gather}
implies $\hat{x}=x$, $\hat{y}=y$, $\hat{z}=z$, then
$Y_{a,b}$ satisfies the parametric Yang--Baxter equation \eqref{ParamYB}.
\end{proposition}

\subsection{Complete integrability}
In the following definition we recall the standard notion of complete (Liouville) integrability for maps on manifolds
(see, e.g.,~\cite{fordy14, vesel1991} and references therein).
\begin{definition}
\label{dli}
Let $\dmp$ be a positive integer and 
 $\Mfd$ a $\dmp$-dimensional manifold 
with (local) coordinates $\xc_1,\dots,\xc_\dmp$.
A (smooth or analytic) map $F\cl \Mfd\to \Mfd$ is said to be \emph{Liouville integrable} 
(or \emph{completely integrable}) if 
one has the following objects on the manifold~$\Mfd$.
\begin{itemize}
	\item A Poisson bracket $\{\,,\,\}$ which is 
	invariant under the map~$F$ and is of constant rank~$2r$ 
	for some positive integer~$r\le\dmp/2$ (i.e. the $\dmp\times\dmp$ matrix with the entries 
	$\{\xc_i,\xc_j\}$ is of constant rank~$2r$).
	The invariance of the bracket means the following.
	For any functions $g$, $h$ on~$\Mfd$ one has $\{g,h\}\circ F=\{g\circ F,\,h\circ F\}$.
To prove that the bracket is invariant,
it is sufficient to check $\{g,h\}\circ F=\{g\circ F,\,h\circ F\}$ for $g=\xc_i$, $\,h=\xc_j$, $\,i,j=1,\dots,\dmp$.
	
\item If $2r<\dmp$, then one needs also $\dmp-2r$ functions $C_s$, $\,s=1,\dots,\dmp-2r$,
which are invariant under~$F$ (i.e. $C_s\circ F=C_s$) 
and are Casimir functions (i.e. $\{C_s,g\}=0$ for any function~$g$).
	\item One has $r$ functions $I_l$, $\,l=1,\dots,r$, which are invariant under~$F$
	and are in involution with respect to the Poisson bracket (i.e. $\{I_{{\rm l}_1},I_{l_2}\}=0$ 
	for all $l_1,l_2=1,\dots,r$).
	\item The functions $C_1,\dots,C_{\dmp-2r},I_1,\dots,I_r$  must be functionally independent.
\end{itemize}	
\end{definition}

\section{Yang--Baxter maps of KdV type}\label{KdV-type maps}
Most examples of Yang--Baxter maps in the literature --- including all the maps of the famous classification lists \cite{ABS-2004, Pap-Tongas-Veselov} --- possessing a Lax matrix satisfy the Lax equation \eqref{eq-Lax} uniquely. 

The first example of Yang--Baxter map with Lax operator that does not uniquely satisfy the Lax equation \eqref{eq-Lax} first appeared in \cite{Sokor-Kouloukas}. Motivated by that,  a new approach for correspondences giving rise to Zamolodchikov tetrahedron maps (which are higher-dimensional analogues of Yang--Baxter maps) was presented in in \cite{Igonin-Sokor}.

In this section, using a KdV type of Yang--Baxter map as an illustrative example, we demonstrate the benefits of the approach presented in \cite{Igonin-Sokor} in the case of Yang--Baxter maps. As a result, we construct new KdV type Yang--Baxter maps.

\subsection{Commutative Yang--Baxter maps of KdV type}
Let $\mathcal{X}=\mathbb{C}-\left\{0\right\}$ and ${\rm L}_a(x_1,x_2)$ be a matrix given by
\begin{equation}\label{YBlaxpKdV}
{\rm L}_a(x_1,x_2):=
\left(
\begin{matrix}
 x_1 & a+x_1x_2-\lambda \\
 -1 & -x_2
\end{matrix}\right), \quad x_1,x_2\in\mathcal{X},~~a,\lambda\in \mathbb{C},
\end{equation}
and consider the matrix refactorisation problem \eqref{eq-Lax}:
\begin{equation}\label{Lax-KdV}
    {\rm L}_a(u_1,u_2){\rm L}_b(v_1,v_2)={\rm L}_b(y_1,y_2){\rm L}_a(x_1,x_2).
\end{equation}

Equation \eqref{Lax-KdV} is equivalent to the following system of polynomial equations
\begin{align*}
    u_1-v_2&=y_1-x_2,\quad u_2-v_1=y_2-x_1, \quad a+u_1(u_2-v_1)=b+y_1(y_2-x_1),\\
     a+v_2(u_2-v_1)&=b+x_2(y_2-x_1),\quad u_1(b-u_2v_2+v_1v_2)-av_2=y_1 (a-y_2x_2+x_1x_2)-bx_2.
\end{align*}
The above can be solved for $u_1,u_2$ and $v_2$ in terms of $v_1$ to obtain:
\begin{subequations}\label{corr-KdV}
\begin{align}
    u_1&=y_1+\frac{a-b}{x_1-y_2},\\
    u_2&=v_1-x_1+y_2,\\
    v_2&=x_2+\frac{a-b}{x_1-y_2}.
\end{align}
\end{subequations}
That is, \eqref{corr-KdV} is not a unique solution of \eqref{Lax-KdV}, but it satisfies \eqref{Lax-KdV} for arbitrary $v_1\in\mathcal{X}$. In other words equation \eqref{Lax-KdV} implies a correspondence between $\mathcal{X}^4$ and $\mathcal{X}^4$, namely relations \eqref{corr-KdV}, rather than a map.

There are two ways to deal with this issue. First, one can introduce an auxiliary parameter $\epsilon$ in ${\rm L}_a(u_1,u_2)$ in \eqref{YBlaxpKdV} so that from the substitution of the corresponding matrix into \eqref{Lax-KdV} we would obtain a map. For example, in \cite{Kouloukas} the authors consider the matrix $\hat{\rm L}(x_1,x_2,a,\epsilon)=\begin{pmatrix}x_1-\epsilon\lambda & \frac{a+x_1x_2}{1-\epsilon (x_1-x_2)}-\lambda \\ -1+\epsilon (x_1-x_2) & -x_2-\epsilon \lambda\end{pmatrix}$, such that $\lim_{\epsilon\rightarrow 0}\hat{\rm L}(x_1,x_2,a,\epsilon)={\rm L}_a(x_1,x_2)$. Then, then the matrix refactorisation problem 
$$
 \hat{\rm L}_a(u_1,u_2,\epsilon)\hat{\rm L}_b(v_1,v_2,\epsilon)=\hat{\rm L}_b(y_1,y_2,\epsilon)\hat{\rm L}_a(x_1,x_2,\epsilon)
$$
is equivalent to a map which tends to the KdV type map \cite{Kouloukas}
\begin{equation}\label{KdV-lift}
    (x_1,x_2,y_1,y_2)\stackrel{Y_{a,b}}{\longrightarrow}\left(y_1+\frac{a-b}{x_1-y_2},y_2,x_1,x_2+\frac{a-b}{x_1-y_2}\right),
\end{equation}
as $\epsilon\rightarrow 0$. We note that matrix $\hat{\rm L}(x_1,x_2,a,\epsilon)$ in \cite{Kouloukas} was not randomly chosen, but follows naturally from certain Casimir functions.

Another way to obtain map \eqref{corr-KdV} is to adopt the correspondence approach which was presented in \cite{Igonin-Sokor} for constructing tetrahedron maps. In particular, employing \eqref{corr-KdV} as a generator of potential Yang--Baxter maps and the choosing the free parameter $v_1$ to be $v_1=x_1$, we obtain map \ref{KdV-lift}. 

We can construct more Yang--Baxter maps from \eqref{corr-KdV}. In particular, we have the following.

\begin{proposition}
    The following maps
    \begin{align}
        Y_1:(x_1,x_2,y_1,y_2)\longrightarrow\left(y_1+\frac{a-b}{x_1-y_2},y_2-\frac{a-b}{x_1-y_2},x_1-\frac{a-b}{x_1-y_2},x_2+\frac{a-b}{x_1-y_2}\right),\label{KdV-1}\\
        Y_2:(x_1,x_2,y_1,y_2)\longrightarrow\left(y_1+\frac{a-b}{x_1-y_2},y_2-\frac{a-b}{x_2-y_1},x_1-\frac{a-b}{x_2-y_1},x_2+\frac{a-b}{x_1-y_2}\right),\label{KdV-2}
    \end{align}
    are noninvolutive parametric Yang--Baxter maps. Moreover, maps $Y_1$ and $Y_2$ share the following functionally independent first integrals
    $$
    I_1=x_1-x_2+y_1-y_2,\quad I_2=(x_2-y_1)(x_1-y_2),
    $$
    and $Y_1$ has one more first integral:
    $$
    I_3=x_1+x_2+y_1+y_2.
    $$
\end{proposition}

\begin{proof}
    Maps $Y_1$ and $Y_2$ follow from the correspondence \eqref{corr-KdV} for $v_1=x_1-\frac{a-b}{x_1-y_2}$ and $v_1=x_1-\frac{a-b}{x_2-y_1}$, respectively. They are both nonivolutive since, for instance, $u_1\circ Y_1=x_1+(a-b)\left(\frac{1}{y_1-x_2}+\frac{1}{y_2-x_1}\right)\neq x_1$ and $u_1\circ Y_2=x_1+\frac{2(a-b)}{y_1-x_2}\neq x_1$. The Yang--Baxter property can be verified by straightforward substitution to the Yang--Baxter equation.

    Regarding the invariants, we have that $u_1-u_2+v_1-v_2\stackbin[\eqref{KdV-1}]{\eqref{KdV-2}}{=}x_1-x_2+y_1-y_2$, $(u_2-v_1)(u_1-v_2)\stackbin[\eqref{KdV-1}]{\eqref{KdV-2}}{=}(x_2-y_1)(x_1-y_2)$, and $u_1+u_2+v_1+v_2\stackrel{\eqref{KdV-1}}{=}x_1+x_2+y_1+y_2$. Now, if $r_i=\nabla I_i$, $i=1,2,3$, then the matrix ${\rm R}=\left[r_1,r_2,r_3\right]$ has rank 3. That is, $I_i$, $i=1,2,3$, are functionally independent. 
\end{proof}

Regarding the Liouville integrability of these maps, we have the following.

\begin{proposition}
    Map \eqref{KdV-1} is completely integrable.
\end{proposition}
\begin{proof} Map \eqref{KdV-1} has three functionally independent first integrals $I_1=x_1-x_2+y_1-y_2$, $I_2=(x_2-y_1)(x_1-y_2)$ and $I_3=x_1+x_2+y_1+y_2$. The Poisson bracket defined by
$$
\{x_1,x_2\}=\{x_2,y_1\}=\{y_1,y_2\}=-1,\quad \{x_1,x_2\}=1, \quad\text{and all the rest are 0},
$$
is invariant under the map \eqref{KdV-1}. The associated Poisson matrix ${\rm P}$ is of rank 2, and $I_1$, $I_2$ are Casimir functions, since $\nabla I_1\cdot{\rm P}=\nabla I_2\cdot{\rm P}=0$. Therefore, map \eqref{KdV-1} is completely integrable.
\end{proof}

\subsection{A noncommutative Yang--Baxter map of KdV type}
Let $\mathfrak{R}$ be a noncommutative division ring and $Z(\mathfrak{R})$ its centre. Consider the following matrix
\begin{equation}\label{YBlaxpKdVNC}
{\rm L}_a(\bm{x}_1,\bm{x}_2):=
\left(
\begin{matrix}
 \bm{x}_1 & a+\bm{x}_1\bm{x}_2-\lambda \\
 -1 & -\bm{x}_2
\end{matrix}\right), \quad \bm{x}_1,\bm{x}_2\in\mathcal{R},~~a,\lambda\in Z(\mathfrak{R}),
\end{equation}
which is the noncommutative version of matrix \eqref{YBlaxpKdVNC}.

We substitute ${\rm L}_a(\bm{x}_1,\bm{x}_2)$ to the Lax equation
\begin{equation}\label{Lax-KdV-NC}
    {\rm L}_a(\bm{u}_1,\bm{u}_2){\rm L}_b(\bm{v}_1,\bm{v}_2)={\rm L}_b(\bm{y}_1,\bm{y}_2){\rm L}_a(\bm{x}_1,\bm{x}_2).
\end{equation}
The above matrix refactorisation problem is equivalent to the following correspondence between $\mathcal{R}^4$ and $\mathcal{R}^4$:
\begin{subequations}\label{KdV-corr-NC}
\begin{align}
    \bm{u}_1&=\bm{y}_1+(a-b)(\bm{x}_1-\bm{y}_2)^{-1},\\
    \bm{v}_1&=\bm{x}_1+\bm{u}_2-\bm{y}_2,\\
    \bm{v}_2&=\bm{x}_2+(a-b)(\bm{x}_1-\bm{y}_2)^{-1}.
\end{align}
\end{subequations}

The correspondence \eqref{KdV-corr-NC} does not define a Yang--Baxter map for arbitrary $\bm{u}_2$. However, for the choice $\bm{u}_2=\bm{y}_2$, \eqref{KdV-corr-NC} defines a Yang--Baxter map. In particular, we have the following.

\begin{theorem} (Noncommutative KdV lift) The  map $Y_{a,b}$ given by
\begin{subequations}\label{KdV-YB-NC}
\begin{align}
   \bm{x}_1\mapsto \bm{u}_1&=\bm{y}_1+(a-b)(\bm{x}_1- \bm{y}_2)^{-1},\label{KdV-YB-NC-a}\\
   \bm{x}_2\mapsto \bm{u}_2&=\bm{y}_2,\label{KdV-YB-NC-b}\\
    \bm{y}_1\mapsto\bm{v}_1&=\bm{x}_1,\label{KdV-YB-NC-c}\\
    \bm{y}_2\mapsto\bm{v}_2&=\bm{x}_2+(a-b)(\bm{x}_1-\bm{y}_2)^{-1},\label{KdV-YB-NC-d}
\end{align}
\end{subequations}
is a noninvolutive, parametric Yang--Baxter map. Furthermore,  map \eqref{KdV-YB-NC}  admits  the  first  integral $I=\bm{x}_1-\bm{x}_2+\bm{y}_1-\bm{y}_2$.
\end{theorem}

\begin{proof}
    Map \eqref{KdV-YB-NC} follows from \eqref{KdV-corr-NC} for the choice $u_2=y_2$.

    For the proof of the Yang--Baxter property we adopt the following notation. We denote the actions of the right- and left-side of the Yang--Baxter equation on $(\bm{x},\bm{y},\bm{z})$ as follows:
         \begin{align*}
         Y^{12}_{a,b}(\bm{x},\bm{y},\bm{z})&=(\tilde{\bm{x}},\tilde{\bm{y}},\bm{z}), & &(\bm{x},\hat{\bm{y}},\hat{\bm{z}})=Y^{23}_{b,c}(\bm{x},\bm{y},\bm{z}),\\
         Y^{13}_{a,c} \circ Y^{12}_{a,b}(\bm{x},\bm{y},\bm{z})&=(\tilde{\tilde{\bm{x}}},\tilde{\bm{y}},\tilde{\bm{z}}), & &(\hat{\bm{x}},\hat{\bm{y}},\hat{\hat{\bm{z}}})=Y^{13}_{a,c} \circ Y^{23}_{b,c}(\bm{x},\bm{y},\bm{z}),\\
       Y^{23}_{b,c}\circ  Y^{13}_{a,c} \circ Y^{12}_{a,b}(\bm{x},\bm{y},\bm{z})&=(\tilde{\tilde{\bm{x}}},\tilde{\tilde{\bm{y}}},\tilde{\tilde{\bm{z}}}), & &(\hat{\hat{\bm{x}}},\hat{\hat{\bm{y}}},\hat{\hat{\bm{z}}})=Y^{12}_{a,b}\circ Y^{13}_{a,c} \circ Y^{23}_{b,c}(\bm{x},\bm{y},\bm{z}).
    \end{align*} 
    That is, the Yang--Baxter equation is satisfied if $(\tilde{\bm{x}},\tilde{\bm{y}},\tilde{\bm{z}})=(\hat{\bm{x}},\hat{\bm{y}},\hat{\bm{z}})$.

    Now, using the right-hand side of the Yang--Baxter equation, and taking into account that $\bm{x_2}\mapsto \bm{y_2}$ and $\bm{y}_1\mapsto \bm{x}_1$, we obtain
\begin{align}
&{\rm L}_c(\bm{z}_1,\bm{z}_2){\rm L}_b(\bm{y}_1,\bm{y}_2){\rm L}_a(\bm{x}_1,\bm{x}_2) = {\rm L}_c(\bm{z}_1,\bm{z}_2){\rm L}_a(\tilde{\bm{x}}_1,\bm{y}_2){\rm L}_b(\bm{x}_1,\tilde{\bm{y}}_2)= \nonumber\\
&{\rm L}_a(\tilde{\tilde{\bm{x}}}_1,\bm{z}_2){\rm L}_c(\tilde{\bm{x}}_1,\tilde{\bm{z}}_2){\rm L}_b(\bm{x}_1,\tilde{\bm{y}}_2)={\rm L}_a(\tilde{\tilde{\bm{x}}}_1, \bm{z}_2){\rm L}_b(\tilde{\bm{x}}_1,\tilde{\bm{z}}_2){\rm L}_c(\bm{x}_1,\tilde{\tilde{\bm{z}}}_2)\label{3-fac-KdV-left}
\end{align}
On the other hand, using the left-hand side of the Yang--Baxter equation, and taking into account that $\bm{x_2}\mapsto \bm{y_2}$ and $\bm{y}_1\mapsto \bm{x}_1$, we have that
\begin{align}
&{\rm L}_c(\bm{z}_1,\bm{z}_2){\rm L}_b(\bm{y}_1,\bm{y}_2){\rm L}_a(\bm{x}_1,\bm{x}_2) = {\rm L}_b(\hat{\bm{y}}_1,\bm{z}_2){\rm L}_c(\bm{y}_1, \hat{\bm{z}}_2){\rm L}_a(\bm{x}_1,\bm{x}_2) =\nonumber\\ 
&{\rm L}_b(\hat{\bm{y}}_1, \bm{z}_2){\rm L}_a(\hat{\bm{x}}_1, \hat{\bm{z}}_2){\rm L}_c(\bm{x}_1, \hat{\hat{\bm{z}}}_2) = {\rm L}_a(\hat{\hat{\bm{x}}}_1, \bm{z}_2,)L(\hat{\bm{x}}_1, \hat{\bm{z}}_2,b){\rm L}_c(\bm{x}_1,\hat{\hat{\bm{z}}}_2).\label{3-fac-KdV-right}
\end{align}

From relations \eqref{3-fac-KdV-left} and \eqref{3-fac-KdV-right} it follows that
$$
{\rm L}_a(\tilde{\tilde{\bm{x}}}_1, \bm{z}_2){\rm L}_b(\tilde{\bm{x}}_1,\tilde{\bm{z}}_2){\rm L}_c(\bm{x}_1,\tilde{\tilde{\bm{z}}}_2)= {\rm L}_a(\hat{\hat{\bm{x}}}_1, \bm{z}_2)L(\hat{\bm{x}}_1, \hat{\bm{z}}_2,b){\rm L}_c(\bm{x}_1,\hat{\hat{\bm{z}}}_2),
$$
which implies the system of equations
\begin{subequations}\label{NC_KdV-lift-corr}
\begin{align} 
   & \tilde{\tilde{\bm{x}}}_1 - \tilde{\bm{z}}_2 = \hat{\hat{\bm{x}}}_1 - \hat{\bm{z}}_2,\label{NC_KdV-lift-corr-a}\\
   & \tilde{\bm{x}}_1 - \tilde{\tilde{\bm{z}}}_2 = \hat{\bm{x}}_1 - \hat{\hat{\bm{z}}}_2,\label{NC_KdV-lift-corr-b}\\
    &\tilde{\bm{x}}_1 \bm{x}_1-\tilde{\bm{x}}_1\tilde{\bm{z}}_2 + \bm{z}_2 \tilde{\bm{z}}_2 = \hat{\bm{x}}_1 \bm{x}_1 -\hat{\bm{x}}_1\hat{\bm{z}}_2 + \bm{z}_2 \hat{\bm{z}}_2,\label{NC_KdV-lift-corr-c}\\
    &\tilde{\tilde{\bm{x}}}_1 \tilde{\bm{x}}_1 -  \tilde{\tilde{\bm{x}}}_1 \bm{z}_2 - \bm{x}_1 \tilde{\tilde{\bm{z}}}_2 - \tilde{\tilde{\bm{x}}}_1 \tilde{\tilde{\bm{z}}}_2 + \tilde{\bm{z}}_2 \tilde{\tilde{\bm{z}}}_2 = \hat{\hat{\bm{x}}}_1 \hat{\bm{x}}_1 -  \hat{\hat{\bm{x}}}_1 \bm{z}_2 - \bm{x}_1 \hat{\hat{\bm{z}}}_2 - \hat{\hat{\bm{x}}}_1 \hat{\hat{\bm{z}}}_2 + \hat{\bm{z}}_2 \hat{\hat{\bm{z}}}_2,\label{NC_KdV-lift-corr-d}\\
    &\tilde{\tilde{\bm{x}}}_1 \tilde{\bm{x}}_1 \bm{x}_1 - \tilde{\tilde{\bm{x}}}_1 \bm{z}_2 \bm{x}_1 - b \tilde{\tilde{\bm{x}}}_1 - \tilde{\tilde{\bm{x}}}_1 \tilde{\bm{x}}_1  \tilde{\bm{z}}_2 + a \tilde{\bm{z}}_2 + \tilde{\tilde{\bm{x}}}_1 \bm{z}_2 \tilde{\bm{z}}_2 =\nonumber\\ &\hat{\hat{\bm{x}}}_1 \hat{\bm{x}}_1 \bm{x}_1 - \hat{\hat{\bm{x}}}_1 \bm{z}_2 \bm{x}_1 - b \hat{\hat{\bm{x}}}_1 - \hat{\hat{\bm{x}}}_1 \hat{\bm{x}}_1  \hat{\bm{z}}_2 + a \hat{\bm{z}}_2 + \hat{\hat{\bm{x}}}_1 \bm{z}_2 \hat{\bm{z}}_2,\label{NC_KdV-lift-corr-e}\\
    &\tilde{\tilde{\bm{x}}}_1 \tilde{\bm{x}}_1 (c + \bm{x}_1 \tilde{\tilde{\bm{z}}}_2) - a(c+ \bm{x}_1 \tilde{\tilde{\bm{z}}}_2) - \tilde{\tilde{\bm{x}}}_1 \bm{z}_2 (c+ \bm{x}_1 \tilde{\tilde{\bm{z}}}_2) - \tilde{\tilde{\bm{x}}}_1 (b \tilde{\tilde{\bm{z}}}_2 + \tilde{\bm{x}}_1 \tilde{\bm{z}}_2 \tilde{\tilde{\bm{z}}}_2) + (a + \tilde{\tilde{\bm{x}}}_1 \bm{z}_2)\tilde{\bm{z}}_2 \tilde{{\tilde{\bm{z}}}}_2=\nonumber\\
    &\hat{\hat{\bm{x}}}_1 \hat{\bm{x}}_1 (c + \bm{x}_1 \hat{\hat{\bm{z}}}_2) - a(c+ \bm{x}_1 \hat{\hat{\bm{z}}}_2) - \hat{\hat{\bm{x}}}_1 \bm{z}_2 (c+ \bm{x}_1 \hat{\hat{\bm{z}}}_2)- \hat{\hat{\bm{x}}}_1 (b \hat{\hat{\bm{z}}}_2 + \hat{\bm{x}}_1 \hat{\bm{z}}_2 \hat{\hat{\bm{z}}}_2) + (a + \hat{\hat{\bm{x}}}_1 \bm{z}_2)\hat{\bm{z}}_2 \hat{\hat{\bm{z}}}_2,\label{NC_KdV-lift-corr-f}\\
    &c \tilde{\bm{x}}_1 + \tilde{\bm{x}}_1 \bm{x}_1 \tilde{\tilde{\bm{z}}}_2 - \bm{z}_2 \bm{x}_1 \tilde{\tilde{\bm{z}}}_2 - b \tilde{\tilde{\bm{z}}}_2 - \tilde{\bm{x}}_1 \tilde{\bm{z}}_2 \tilde{\tilde{\bm{z}}}_2 + \bm{z}_2 \tilde{\bm{z}}_2 \tilde{\tilde{\bm{z}}}_2 = \nonumber\\
    & c \hat{\bm{x}}_1 + \hat{\bm{x}}_1 \bm{x}_1 \hat{\hat{\bm{z}}}_2 - \bm{z}_2 \bm{x}_1 \hat{\hat{\bm{z}}}_2 - b \hat{\hat{\bm{z}}}_2 - \hat{\bm{x}}_1 \hat{\bm{z}}_2 \hat{\hat{\bm{z}}}_2 + \bm{z}_2 \hat{\bm{z}}_2 \hat{\hat{\bm{z}}}_2.\label{NC_KdV-lift-cor-g}
\end{align}
\end{subequations}
Equation \eqref{NC_KdV-lift-cor-g} can be written in the following form
\begin{equation}\label{eq-z2}
    c(\tilde{\bm{x}}_1 -\hat{\bm{x}}_1) + (\tilde{\bm{x}}_1 \bm{x}_1  - \tilde{\bm{x}}_1 \tilde{\bm{z}}_2 + \bm{z}_2 \tilde{\bm{z}}_2)\tilde{\tilde{\bm{z}}}_2 - b \tilde{\tilde{\bm{z}}}_2 - \bm{z}_2 \bm{x}_1 \tilde{\tilde{\bm{z}}}_2 - \hat{\bm{x}}_1 \bm{x}_1 \hat{\hat{\bm{z}}}_2 + \bm{z}_2 \bm{x}_1 \hat{\hat{\bm{z}}}_2 + b \hat{\hat{\bm{z}}}_2 + \hat{\bm{x}}_1 \hat{\bm{z}}_2 \hat{\hat{\bm{z}}}_2 - \bm{z}_2 \hat{\bm{z}}_2 \hat{\hat{\bm{z}}}_2 = 0.
\end{equation}

Now, with the use of equations \eqref{NC_KdV-lift-corr-b} and \eqref{NC_KdV-lift-corr-c}, equation \eqref{eq-z2} takes the form
$$
c(\tilde{\tilde{\bm{z}}}_2-\hat{\hat{\bm{z}}}_2) + (\hat{\bm{x}}_1 \bm{x}_1 -\hat{\bm{x}}_1\hat{\bm{z}}_2 + \bm{z}_2 \hat{\bm{z}}_2)\tilde{\tilde{\bm{z}}}_2 - b \tilde{\tilde{\bm{z}}}_2 - \bm{z}_2 \bm{x}_1 \tilde{\tilde{\bm{z}}}_2 - \hat{\bm{x}}_1 \bm{x}_1 \hat{\hat{\bm{z}}}_2 + \bm{z}_2 \bm{x}_1 \hat{\hat{\bm{z}}}_2 + b \hat{\hat{\bm{z}}}_2 + \hat{\bm{x}}_1 \hat{\bm{z}}_2 \hat{\hat{\bm{z}}}_2 - \bm{z}_2 \hat{\bm{z}}_2 \hat{\hat{\bm{z}}}_2 = 0,
$$
or
$$ (\hat{\hat{\bm{z}}}_2 - \tilde{\tilde{\bm{z}}}_2) (c - b + \hat{\bm{x}}_1 \bm{x}_1 - \bm{z}_2 \bm{x}_1 - \hat{\bm{x}}_1 \hat{\bm{z}}_2 + \bm{z}_2 \hat{\bm{z}}_2) = 0,
$$
from which follows that $\hat{\hat{\bm{z}}}_2 = \tilde{\tilde{\bm{z}}}_2$.

Given that $\hat{\hat{\bm{z}}}_2 = \tilde{\tilde{\bm{z}}}_2$, equation \eqref{NC_KdV-lift-corr-b} implies that $\hat{\bm{x}}_1=\tilde{\bm{x}}_1$. After substitution of the latter to \eqref{NC_KdV-lift-corr-c}, it follows that $(\hat{\bm{x}}_1 - \bm{z}_2)(\tilde{\bm{z}}_2 - \hat{\bm{z}}_2) = 0$, i.e. $\tilde{\bm{z}}_2 =\hat{\bm{z}}_2$, which implies $\tilde{\tilde{\bm{x}}}_1 = \hat{\hat{\bm{x}}}_1$. Thus, map \eqref{KdV-YB-NC} is a parametric Yang--Baxter map.

The noninvolutivity of this map can be proved, for instance, by the fact that: $u_1\circ Y_{a,b}=x_1+(a-b)(y_1-x_2)^{-1}\neq x_1$, thus $Y_{a,b}\circ Y_{a,b}\neq\id$.

Regarding the first  integral, one can verify that   $\bm{u}_1-\bm{u}_2+\bm{v}_1-\bm{v}_2\stackrel{\eqref{KdV-YB-NC}}{=\bm{x}_1+\bm{y}_1-\bm{x}_2-\bm{y}_2}$, by  straightforward  substitution.
\end{proof}

\begin{remark}\normalfont
    Although map \eqref{KdV-YB-NC} admits the first integral $I_1=\bm{x}_1+\bm{y}_1-\bm{x}_2-\bm{y}_2$,  it lacks the integral $I_2=(\bm{x}_2-\bm{y}_1)(\bm{x}_1-\bm{y}_2)$ that its commutative version \eqref{KdV-lift} possesses. 
\end{remark}

\subsection{Squeeze down to the noncommutative discrete potential KdV lattice}
Recall that an equation on quad-graph \cite{ABS-2004} is an a difference equation of the form
\begin{equation}\label{quad-graph}
Q(f_{n,m},f_{n+1,m},f_{n,m+1},f_{n+1,m+1};a,b),
\end{equation}
where $Q$ is polynomial affine linear in $f_{n,m},f_{n+1,m},f_{n,m+1}$ and $f_{n+1,m+1}$. Probably, the most popular example of quad-graph equation is the discrete potential KdV equation
\begin{equation}\label{dpKdV}
    (f_{n+1,m+1}-f_{n,m}) (f_{n,m+1}-f_{n+1,m})=a-b.
\end{equation}

Integrability of a quad-graph equation mean that there is a pair of matrices ${\rm M}={\rm M}(f_{n,m},f_{n+1,m},a)$ and ${\rm A}={\rm A}(f_{n,m},f_{n,m+1},b)$ such that equation \eqref{quad-graph} is equivalent to the following Lax equation:
\begin{equation}\label{H1}
{\rm M}(f_{n,m+1},f_{n+1,m+1},a){\rm A}(f_{n,m},f_{n,m+1},b)={\rm A}(f_{n+1,m},f_{n+1,m+1},b){\rm }{\rm M}(f_{n,m},f_{n+1,m},a).
\end{equation}

We aim to relate the noncommutative Yang--Baxter map \eqref{KdV-YB-NC} to the \eqref{H1} similar to the commutative case \cite{Kouloukas}. As we saw in the previous section, the noncommutative map \eqref{KdV-YB-NC} preserves the Yang--Baxter property. Moreover, it shares with its commutative version\eqref{KdV-lift} the following property
$$
\text{If}\quad \bm{x}_2=\bm{y}_1,\quad \text{then}\quad \bm{u}_1=\bm{v}_2.
$$
This will allow us to squeeze map \eqref{KdV-YB-NC} to the noncommutative discrete potential KdV equation.

In particular, we have the following. 

\begin{proposition}
    The Yang--Baxter map \eqref{KdV-YB-NC} squeezes down to the noncommutative discrete potential KdV equation:
    \begin{equation}\label{dpKdV-NC}
    ({\bm f}_{n+1,m+1}-{\bm f}_{n,m}) ({\bm f}_{n,m+1}-{\bm f}_{n+1,m})=a-b,
\end{equation}
which admits the following Lax representation
\begin{equation}\label{H1-NC}
{\rm L}({\bm f}_{n,m+1},{\bm f}_{n+1,m+1},a){\rm L}({\bm f}_{n,m},{\bm f}_{n,m+1},b)={\rm L}({\bm f}_{n+1,m},{\bm f}_{n+1,m+1},b){\rm }{\rm L}({\bm f}_{n,m},{\bm f}_{n+1,m},a),
\end{equation}
where ${\rm L}={\rm L}({\bm f}_{n,m},{\bm f}_{n+1,m},a)=\begin{pmatrix} {\bm f}_{n,m} & a+{\bm f}_{n,m}{\bm f}_{n+1,m}\\ -1 & - {\bm f}_{n+1,m}\end{pmatrix}$.
\end{proposition}
\begin{proof}
    From \eqref{KdV-YB-NC-a} and \eqref{KdV-YB-NC-d} follows that if ${\bm y}_1={\bm x}_2$, then ${\bm u}_1={\bm v}_2$. This fact as well as comparing the Lax representations \eqref{Lax-KdV-NC} and \eqref{H1-NC} of \eqref{KdV-YB-NC} and \eqref{dpKdV-NC}, respectively, indicates to set ${\bm u}_1={\bm v}_2={\bm f}_{n,m+1}$, ${\bm u}_2={\bm y}_2={\bm f}_{n+1,m+1}$, ${\bm v}_1={\bm x}_1={\bm f}_{n,m}$ and ${\bm y}_1={\bm x}_2={\bm f}_{n+1,m}$. After this change of variables, \eqref{KdV-YB-NC-b} and \eqref{KdV-YB-NC-c} are identically satisfied, whereas \eqref{KdV-YB-NC-a} and \eqref{KdV-YB-NC-d} imply
     $$
    {\bm f}_{n,m+1}-{\bm f}_{n+1,m}=(a-b)({\bm f}_{n,m}-{\bm f}_{n+1,m+1})^{-1},
    $$
    which can be written as \eqref{dpKdV-NC}, if we multiply by ${\bm f}_{n,m}-{\bm f}_{n+1,m+1}$ from the right.
\end{proof}

\section{Darboux--B\"acklund transformations and NLS type Yang--Baxter maps on division rings}\label{NLS type  maps}
Following \cite{SPS} we define a Darboux transformation as a similarity-type transformation that leaves covariant a Lax operator $\mathfrak{L}=D_x-{\rm U}(\bm{p},\bm{q};\lambda)$. That is, a transformation
\begin{equation}\label{D-transform}
    \mathfrak{L}(\bm{p},\bm{q};\lambda)\rightarrow\mathfrak{L}(\tilde{\bm{p}},\tilde{\bm{q}};\lambda)={\rm M}\mathfrak{L}(\bm{p},\bm{q};\lambda){\rm M}^{-1}.
\end{equation}

In order to find ${\rm M}$ we need to assume an initial form for it. We usually start with the simple case of ${\rm M}$ being linear in the spectral parameter $\lambda$, i.e. ${\rm M}=\lambda {\rm M}^1+{\rm M}^0$.

\subsection{Noncommutative NLS type Darboux transformations}
A Darboux transformation was presented in \cite{SP} for the noncommutative NLS operator 
$$\mathcal{L}_{NLS}=D_x-\lambda\begin{pmatrix} 1 & 0 \\ 0 & -1\end{pmatrix}-\begin{pmatrix} 0 & 2\bm{p} \\ 2\bm{q} & 0\end{pmatrix}
$$
where $\bm{p}$, $\bm{q}$ belong to some noncommutative division ring $\mathfrak{R}$.
The associated Darboux matrix reads:
\begin{equation}\label{DT-NLS}
{\rm M} =\lambda\begin{pmatrix}
1 & 0 \\
0 & 0
\end{pmatrix}
+
\begin{pmatrix}
a+\bm{p}\tilde{\bm{q}} & \bm{p} \\
\tilde{\bm{p}} & 1
\end{pmatrix},\quad \bm{p},\bm{q},\tilde{\bm{p}}\tilde{\bm{q}}\in \mathfrak{R},\quad a\in Z(\mathfrak{R}).
\end{equation}

Here, we construct a Darboux transformation for the noncommutative derivative NLS operator 
$$
    \mathcal{L}_{DNLS}=D_x-\lambda^2 {\rm U}^2-\lambda {\rm U}^1 =D_x-\lambda^2\begin{pmatrix} 1 & 0 \\ 0 & -1\end{pmatrix}-\lambda\begin{pmatrix} 0 & 2\bm{p} \\ 2\bm{q} & 0\end{pmatrix},
$$
which is invariant under the transformation
\begin{equation}\label{symmetry}
    s_1(\lambda):\mathcal{L}_{DNLS}(\lambda)\rightarrow \sigma_3 \mathcal{L}_{DNLS}(-\lambda)\sigma_3,
\end{equation}
where $\sigma_3$ is the standard Pauli matrix $\sigma_3=\begin{pmatrix}1 & 0\\ 0 & -1\end{pmatrix}$. 

Similarly to the commutative case \cite{SPS}, we seek Darboux matrices of the form
$$
\rm{M}=\lambda^2 \rm{M}^2+\lambda \rm{M}^1+\rm{M}_0=\lambda^2\begin{pmatrix}\bm{\alpha} & \bm{\beta}\\ \bm{\gamma} & \bm{\delta} \end{pmatrix}+\lambda\begin{pmatrix}
\bm{\kappa} & \bm{\mu} \\ \bm{\nu} & \bm{\xi} \end{pmatrix}+ \begin{pmatrix}
\bm{\rho} & \bm{\sigma} \\ \bm{\theta} & \bm{\phi} \end{pmatrix}.
$$
However, the above matrix must share the same symmetry \eqref{symmetry} with the Lax operator $\mathcal{L}_{DNLS}$, i.e. we demand that 
$$
\rm{M}(\lambda)\rightarrow \sigma_3 \rm{M}(-\lambda)\sigma_3.
$$
Due to the above symmetry, we have that $\bm{\beta}=\bm{\gamma}=\bm{\kappa}=\bm{\xi}=\bm{\sigma}=\bm{\theta}=0$. Thus, eventually, the Darboux matrix we are looking for has  the following form:
\begin{equation}\label{M-gen-form}
\rm{M}=\lambda^2\begin{pmatrix}\bm{\alpha} & 0\\ 0 & \bm{\delta} \end{pmatrix}+\lambda\begin{pmatrix}
0 & \bm{\mu} \\ \bm{\nu} & 0 \end{pmatrix}+ \begin{pmatrix}
\bm{\rho} & 0 \\ 0 & \bm{\phi} \end{pmatrix}.
\end{equation}
For simplicity, following \cite{SPS}, we seek Darboux matrices for which $\rank \rm{M}^2=1$. 

Specifically, we have the following.

\begin{proposition}
Let $\bm{f}$ and $\bm{g}$ be elements of a noncommutative division ring $\mathfrak{R}$ that commute with elements $\bm{p}\bm{q}_{10}$ and $\bm{q}\bm{p}_{10}$, respectively, i.e. $\bm{f}(\bm{p}\bm{q}_{10})=(\bm{p}\bm{q}_{10})\bm{f}$ and $\bm{g}(\bm{q}\bm{p}_{10})=(\bm{q}\bm{p}_{10})\bm{g}$. Then, all the possible quadratic in $\lambda$ Darboux transformations, ${\rm M}=\lambda^2 {\rm M}^2+\lambda {\rm M}^1+{\rm M}^0$, with $\rank {\rm M}^2=1$, are the following.
\begin{enumerate}
\item Matrix 
\begin{equation}\label{DT-DNLS-1}
{\rm M} =\lambda^2\begin{pmatrix} \bm{f} & 0\\ 0 & 0\end{pmatrix}+\lambda \begin{pmatrix} 0 & \bm{f}\bm{p}\\ \tilde{\bm{q}}\bm{f} & 0 \end{pmatrix}+\begin{pmatrix} c_1 & 0\\ 0 & c_2\end{pmatrix}
\end{equation}
where $\bm{f}$, $\bm{p}$ and $\bm{q}$ satisfy the system of differential equations:
\begin{align}\label{BT-DNLS-1}
\bm{f}_x=2(\tilde{\bm{p}}\tilde{\bm{q}}\bm{f}-\bm{f}\bm{p}\bm{q}),\quad
    (\bm{f}\bm{p})_x=2c_2\tilde{\bm{p}}-2c_1\bm{p},\quad
    (\tilde{\bm{q}}\bm{f})_x=2c_1\tilde{\bm{q}}-2c_2\bm{q}.
\end{align}

\item Matrix 
\begin{equation}\label{DT-DNLS-2}
{\rm M} =\lambda^2\begin{pmatrix} 0 & 0\\ 0 & \bm{g}\end{pmatrix}+\lambda \begin{pmatrix} 0 & -\tilde{\bm{p}}\bm{g}\\ -\bm{g}\tilde{\bm{q}} & 0 \end{pmatrix}+\begin{pmatrix} c_1 & 0\\ 0 & c_2\end{pmatrix}
\end{equation}
where $\bm{f}$, $\bm{p}$ and $\bm{q}$ satisfy the system of differential equations:
\begin{align}\label{BT-DNLS-2}
\bm{g}_x=2(\bm{g}\bm{q}\bm{p}\bm{f}-\tilde{\bm{q}}\tilde{\bm{p}}\bm{g}),\quad
    (\tilde{\bm{p}}\bm{g})_x=2c_1\bm{p}-2c_2\tilde{\bm{p}},\quad
    (\bm{g}\bm{q})_x=2c_2\bm{q}-2c_1\tilde{\bm{q}}.
\end{align}
\end{enumerate}
\end{proposition}

\begin{proof}
    We seek a matrix of the form \eqref{M-gen-form} such that
    $$
    \mathfrak{L}_{DNLS}(\tilde{\bm{p}},\tilde{\bm{q}};\lambda){\rm M}={\rm M}\mathfrak{L}_{DNLS}(\bm{p},\bm{q};\lambda).
    $$
    The above implies the following system of polynomial equations:
    \begin{subequations}\label{DT-def}
        \begin{align}
        &\lambda^4: \quad {\rm U}^2 {\rm M}^2 = {\rm M}^2 {\rm U}^2,\label{DT-def-a}\\
        &\lambda^3:\quad {\rm U}^2{\rm M}^1+\tilde{{\rm U}}^1{\rm M}^2={\rm M}^2{\rm U}^1+{\rm M}^1{\rm U}^2,\label{DT-def-b}\\
        &\lambda^2:\quad{\rm M}^2_x-{\rm U}^2{\rm M}^0-\tilde{{\rm U}}^1{\rm M}^1=-{\rm M}^1{\rm U}^1-{\rm M}^0{\rm U}^2,\label{DT-def-c}\\
        &\lambda^1:\quad {\rm M}^1_x-\tilde{{\rm U}}^1{\rm M}^0=-{\rm M}^0{\rm U}^1,\label{DT-def-d}\\
        &\lambda^0:\quad {\rm M}^0_x=0.\label{DT-def-e}
        \end{align}
    \end{subequations}

    Equation \eqref{DT-def-a} is identically satisfied, whereas equation \eqref{DT-def-e} implies that $\bm{\rho}$ and $\bm{\phi}$ are constants, let $\bm{\rho}=c_1$ and $\bm{\phi}=c_2$. Moreover, from \eqref{DT-def-b} it follows that
    \begin{equation}\label{nu-mu}
        \bm{\mu}=\bm{\alpha}\bm{p}-\tilde{\bm{p}}\bm{\delta}, \quad \bm{\nu}=\tilde{\bm{q}}\bm{\alpha}-\bm{\delta}\bm{q}.
    \end{equation}

    Now, $\rank {\rm M}^2=1$ implies that one of $\bm{\alpha}$ and $\bm{\delta}$ is zero and the other is an arbitrary function. Let $\bm{\delta}=0$ and $\bm{\alpha}=\bm{f}(x)$. Then, from \eqref{nu-mu} we obtain that $ \bm{\mu}=\bm{\alpha}\bm{p}$ and $\bm{\nu}=\tilde{\bm{q}}\bm{\alpha}$. Furthermore, equation \eqref{DT-def-c} implies that $\bm{f}(x)$ must satisfy the system of differential equations \eqref{BT-DNLS-1}.

    Now, if $\bm{\alpha}=0$ and $\bm{\delta}=\bm{g}(x)$, then we can similarly prove that $\rm{M}$ is given by \eqref{DT-DNLS-2}, and its entries obey the system of differential equations \eqref{BT-DNLS-2}.
\end{proof}

The  system of differential  equations \eqref{DT-DNLS-1} admits  the following  first  integral:
$$
\partial_x(\bm{f}\bm{p}\tilde{\bm{q}}\bm{f}-c_2\bm{f})=0.
$$
Indeed,
\begin{align*}
    \partial_x(\bm{f}\bm{p}\tilde{\bm{q}}\bm{f}-c_2\bm{f})&=(\bm{fp})_x\bm{\tilde{q}}\bm{f}+\bm{f}\bm{p}(\tilde{\bm{q}}\bm{f})_x-c_2\bm{f}_x\stackrel{\eqref{DT-DNLS-1}}{=}0.
\end{align*}
Therefore, $\bm{f}\bm{p}\tilde{\bm{q}}\bm{f}-c_2\bm{f}=-a=$const., or, solving for $\bm{f}$:
$$
\bm{f}=\frac{a}{c_2}+\frac{1}{c_2}\bm{f}\bm{p}\tilde{\bm{q}}\bm{f}.
$$
We replace $\bm{f}$ by the above expression in \eqref{DT-DNLS-1}, and it follows that the following matrix
\begin{equation}\label{DT-DNLS-1-pq}
{\rm M} =\lambda^2\begin{pmatrix} \frac{a}{c_2}+\frac{1}{c_2}\bm{f}\bm{p}\tilde{\bm{q}}\bm{f} & 0\\ 0 & 0\end{pmatrix}+\lambda \begin{pmatrix} 0 & \bm{f}\bm{p}\\ \tilde{\bm{q}}\bm{f} & 0 \end{pmatrix}+\begin{pmatrix} c_1 & 0\\ 0 & c_2\end{pmatrix},
\end{equation}
is  a Darboux  matrix  for the derivative NLS operator $\mathcal{L}_{DNLS}$.

\subsection{Noncommutative Adler--Yamilov map}
Let $\mathfrak{R}$ be a noncommutative division ring. We change $(\bm{p},\tilde{\bm{q}})\rightarrow (\bm{x}_1,\bm{x_2})$ in \eqref{DT-NLS}, namely we consider the following matrix
\begin{equation}\label{Lax-NLS}
{\rm M}(\bm{x}_1,\bm{x}_2;a) =\lambda\begin{pmatrix}
1 & 0 \\
0 & 0
\end{pmatrix}
+
\begin{pmatrix}
a+\bm{x}_1\bm{x}_2 & \bm{x}_1 \\
\bm{x}_2 & 1
\end{pmatrix},\quad \bm{x}_1,\bm{x}_2\in\mathfrak{R},\quad\lambda\in\mathbb{C},\quad a\in Z(\mathfrak{R}).
\end{equation}

We substitute ${\rm M}$ to the matrix refactorisation problem \eqref{eq-Lax} in order to generate a noncommutative Yang--Baxter map. In particular, we have the following. 

\begin{proposition} (Noncommutative Adler--Yamilov) Let $\mathfrak{A}=\mathfrak{R}\times\mathfrak{R}\times Z(\mathfrak{R})$. The map $Y_{a,b}:\mathfrak{A}^2 \rightarrow \mathfrak{A}^2$, given by
\begin{subequations}\label{n_nls}
    \begin{align}
        \bm{x}_1\mapsto \bm{u}_1 &= \bm{y}_1-(a-b) \bm{x}_1 (1+\bm{y}_2 \bm{x}_1)^{-1},\\
        \bm{x}_2\mapsto \bm{u}_2 &=\bm{y}_2,\\
        \bm{y}_1\mapsto \bm{v}_1 &=\bm{x}_1,\\
        \bm{y}_2\mapsto \bm{v}_2 &=\bm{x}_2+(a-b)(1+\bm{y}_2\bm{x}_1)^{-1} \bm{y}_2,
    \end{align}
\end{subequations}
has the following Lax representation
\begin{equation}\label{Lax-Adler-Yamilov}
    {\rm M}(\bm{u}_1,\bm{u}_2;a){\rm M}(\bm{v}_1,\bm{v}_2;b)={\rm M}(\bm{y}_1,\bm{y}_2;b){\rm M}(\bm{x}_1,\bm{x}_2;a),
\end{equation}
where $\rm{M}$ is given by \eqref{Lax-NLS}, and admits the following first integral
\begin{equation}\label{1st_int-NLS}
I=\bm{x}_1\bm{x}_2+\bm{y}_1\bm{y}_2.
\end{equation}
\end{proposition}
\begin{proof}
    The matrix refactorisation problem \eqref{Lax-Adler-Yamilov} is equivalent to $\bm{u}_2=\bm{y}_2$, $\bm{v}_1=\bm{x}_1$ and the system of polynomial equations
    \begin{subequations}\label{NLS-sys}
    \begin{align}
    \bm{u}_1\bm{y}_2 +\bm{x}_1\bm{v}_2&=\bm{y}_1\bm{y}_2+\bm{x}_1\bm{x}_2,\label{NLS-sys-a}\\
    (a+\bm{u}_1\bm{y}_2)(b+\bm{x}_1\bm{v}_2)+\bm{u}_1\bm{v}_2 &=  (b+\bm{y}_1\bm{y}_2)(a+\bm{x}_1\bm{x}_2)+\bm{y}_1\bm{x}_2,\label{NLS-sys-b}\\
    (a+\bm{u}_1\bm{y}_2)\bm{x}_1+\bm{u}_1 &=  (b+\bm{y}_1\bm{y}_2)\bm{x}_1+\bm{y}_1,\label{NLS-sys-c}\\
    \bm{y}_2(b+\bm{x}_1\bm{v}_2)+\bm{v}_2 &= \bm{y}_2(a+\bm{x}_1\bm{x}_2)+\bm{x}_2\label{NLS-sys-d}.
    \end{align}
\end{subequations}

    Equations \eqref{NLS-sys-c} and \eqref{NLS-sys-d} can be written as $\bm{u}_1(1+\bm{y}_2\bm{x}_1)=(b-a)\bm{x}_1+\bm{y}_1(1+\bm{y}_2\bm{x}_1)=(b-a)\bm{x}_1$ and $(1+\bm{y}_2\bm{x}_1)\bm{v}_2=(a-b)\bm{y}_2+(1+\bm{y}_2\bm{x}_1)\bm{x}_2$. Multiplying these equations by $(1+\bm{y}_2\bm{x}_1)^{-1}$ from the left and the right, respectively, we obtain
    $$
    \bm{u}_1=\bm{y}_1-(a-b) \bm{x}_1 (1+\bm{y}_2 \bm{x}_1)^{-1}\quad\text{and}\quad\bm{v_2}=\bm{x}_2+(a-b)(1+\bm{y}_2\bm{x}_1)^{-1} \bm{y}_2
    $$
    It can be readily verified that $\bm{u}_1$ and $\bm{v}_2$ given by the latter relations satisfy equations \eqref{NLS-sys-a} and \eqref{NLS-sys-b}.

   Finally, $\bm{u}_1\bm{u}_2+\bm{v}_1\bm{v}_2\stackrel{\eqref{n_nls}}{=\bm{x}_1\bm{x}_2+\bm{y}_1\bm{y}_2}$, thus \eqref{1st_int-NLS} is a first  integral of \eqref{n_nls}.
\end{proof}

\begin{remark}\normalfont Map \eqref{n_nls} is the noncommutative version of the Adler--Yamilov Yang--Baxter map \cite{Sokor-Sasha, Kouloukas, Pap-Tongas}. Map \eqref{n_nls} can also be seen in it equivalent form
\begin{align*}
        \bm{x}_1\mapsto \bm{u}_1 &= \bm{y}_1-(a-b) \bm{x}_1 (1+\bm{y}_2 \bm{x}_1)^{-1},\\
        \bm{x}_2\mapsto \bm{u}_2 &=\bm{y}_2,\\
        \bm{y}_1\mapsto \bm{v}_1 &=\bm{x}_1,\\
        \bm{y}_2\mapsto \bm{v}_2 &=\bm{x}_2+(a-b)(1+\bm{y}_2\bm{x}_1)^{-1} \bm{y}_2,
    \end{align*}
\end{remark}

The natural question arises as to whether the Adler--Yamilov map preserves the Yang--Baxter property in the noncommutative case. Specifically, we have the following.

\begin{theorem}
    Map \eqref{n_nls} is a Yang--Baxter map.
\end{theorem}

\begin{proof}
    For matrix \eqref{DT-NLS} we consider the following  
    $$
    {\rm M}(\bm{u}_1,\bm{u}_2;a){\rm M}(\bm{v}_1,\bm{v}_2;b){\rm M}(\bm{z}_1,\bm{z}_2;c)={\rm M}(\bm{z}_1,\bm{x}_2;a){\rm M}(\bm{y}_1,\bm{y}_2;b){\rm M}(\bm{z}_1,\bm{z}_2;c).
    $$

    The above matrix trifactorisation problem is equivalent to $\bm{u}_2=\bm{x}_2$, $\bm{w}_1=\bm{z}_1$ and the system of equations
\begin{subequations}\label{alg-var}
    \begin{align}
        &\bm{z}_1\bm{w}_2 + \bm{u}_1\bm{x}_2 + \bm{v}_1\bm{v}_2 = \bm{z}_1\bm{z}_2 + \bm{x}_1\bm{x}_2 + \bm{y}_1\bm{y}_2,\label{alg-var-a}\\
        &\bm{u}_1\bm{v}_2 + \bm{v}_1\bm{w}_2+\bm{v}_1\bm{v}_2(c+a +\bm{z}_1\bm{w}_2) + (b+a)\bm{z}_1\bm{w}_2  + \bm{u}_1\bm{x}_2(b + c+ \bm{v}_1\bm{v}_2  + \bm{z}_1\bm{w}_2) =\nonumber \\
        & \bm{x}_1\bm{y}_2 + \bm{y}_1\bm{z}_2+\bm{y}_1\bm{y}_2(c+a +\bm{z}_1\bm{z}_2) + (b+a)\bm{z}_1\bm{z}_2  + \bm{x}_1\bm{x}_2(b + c+ \bm{y}_1\bm{y}_2  + \bm{z}_1\bm{z}_2),\label{alg-var-b}\\
       & (\bm{u}_1+a\bm{v}_1)(c\bm{v}_2+\bm{w}_2 ) + \bm{u}_1\bm{x}_2(b + \bm{v}_1\bm{v}_2 )(c+\bm{z}_1\bm{w}_2)+ a(b+\bm{v}_1\bm{v}_2)\bm{z}_1\bm{w}_2   + \bm{u}_1(\bm{v}_2\bm{z}_1+\bm{x}_2\bm{v}_1)\bm{w}_2  =\nonumber \\
       &(\bm{x}_1+a\bm{y}_1)(c\bm{y}_2+\bm{z}_2 ) + \bm{x}_1\bm{x}_2(b + \bm{y}_1\bm{y}_2 )(c+\bm{z}_1\bm{z}_2)+ a(b+\bm{y}_1\bm{y}_2)\bm{z}_1\bm{z}_2   + \bm{x}_1(\bm{y}_2\bm{z}_1+\bm{x}_2\bm{y}_1)\bm{z}_2,\label{alg-var-c}\\
       & \bm{u}_1\bm{x}_2\bm{z}_1 + \bm{v}_1(1+\bm{v}_2\bm{z}_1) = \bm{x}_1\bm{x}_2\bm{z}_1 + \bm{y}_1(1+\bm{y}_2\bm{z}_1)\label{alg-var-d}\\
       &  b\bm{u}_1\bm{x}_2 \bm{z}_1 + (\bm{u}_1+ \bm{u}_1\bm{x}_2\bm{v}_1+a \bm{v}_1) (1+ \bm{v}_2 \bm{z}_1) = \bm{x}_1\bm{x}_2 b \bm{z}_1 + (\bm{x}_1+ \bm{x}_1\bm{x}_2\bm{y}_1+a \bm{y}_1) (1+ \bm{y}_2\bm{z}_1),\label{alg-var-e}\\
       &\bm{x}_2(\bm{z}_1\bm{w}_2 + \bm{v}_1\bm{v}_2) + \bm{v}_2 = \bm{x}_2(\bm{z}_1\bm{z}_2 + \bm{y}_1\bm{y}_2) + \bm{y}_2,\label{alg-var-f}\\
       & (1+\bm{x}_2\bm{v}_1) \bm{v}_2(c + \bm{z}_1\bm{w}_2) + (1+b\bm{x}_2 \bm{z}_1+\bm{x}_2\bm{v}_1)\bm{w}_2=\nonumber \\
       &(1+\bm{x}_2\bm{y}_1) \bm{y}_2(c + \bm{z}_1\bm{z}_2) + (1+b\bm{x}_2 \bm{z}_1+\bm{x}_2\bm{y}_1)\bm{z}_2,\label{alg-var-g}\\
       &\bm{x}_2\bm{v}_1\bm{v}_2\bm{z}_1 + \bm{v}_2 \bm{z}_1 + \bm{x}_2\bm{v}_1 = \bm{x}_2\bm{y}_1\bm{y}_2\bm{z}_1 + \bm{y}_2\bm{z}_1 + \bm{x}_2\bm{y}_1,\label{alg-var-h}
    \end{align}
    \end{subequations}
    for $\bm{u}_1, \bm{v}_1, \bm{v}_2$  and $\bm{w}_2$.

From  \eqref{alg-var-h} we obtain $\bm{v}_1(1+\bm{v}_2\bm{z}_1) = \bm{y}_1\bm{y}_2\bm{z}_1 + \bm{x}_2^{-1}\bm{y}_2\bm{z}_1 + \bm{y}_1- \bm{x}_2^{-1}\bm{v}_2 \bm{z}_1$. Substitution of the latter to  \eqref{alg-var-d} implies $\bm{u}_1\bm{x}_2\bm{z}_1 + \bm{x}_2^{-1}\bm{y}_2\bm{z}_1 - \bm{x}_2^{-1}\bm{v}_2\bm{z}_1 = \bm{x}_1\bm{x}_2\bm{z}_1$, or, after multiplication by $\bm{x}_2$ and $\bm{z}_1^{-1}$ from the left and  right, respectively:
\begin{equation}\label{v2}
	\bm{v}_2 = \bm{x}_2(\bm{u}_1 - \bm{x}_1) \bm{x}_2 + \bm{y}_2.
 \end{equation}

Now, we  substitute $\bm{v}_2$ from \eqref{v2} to  \eqref{alg-var-d} in order to express $\bm{v}_1$ in terms of $\bm{u}_1$, namely we obtain:
\begin{equation}\label{v1}
	\bm{v}_1 = ((\bm{x}_1-\bm{u}_1)\bm{x}_2\bm{z}_1 + \bm{y}_1\bm{y}_2\bm{z}_1 + \bm{y}_1)(\bm{x}_2(\bm{u}_1 - \bm{x}_1) \bm{x}_2\bm{z}_1 + \bm{y}_2\bm{z}_1 + 1)^{-1}.
	\end{equation}

 Next, we solve \eqref{alg-var-h} for $\bm{v}_1\bm{v}_2$,  and we obtain:
 \begin{equation}\label{v1v2}
 \bm{v}_1\bm{v}_2 = \bm{y}_1\bm{y}_2 +\bm{x}_2^{-1}\bm{y}_2 + \bm{y}_1\bm{z}_1^{-1} - \bm{x}_2^{-1}\bm{v}_2 - \bm{v}_1\bm{z}_1^{-1}
 \end{equation}.
 
 Moreover, we substitute \eqref{v2}, \eqref{v1}  and \eqref{v1v2}  to \eqref{alg-var-f}, and  we obtain 
 \begin{equation}\label{w2}
	\bm{w}_2 = \bm{z}_2 + \bm{z}_1^{-1} [(\bm{y}_1\bm{x}_2 + 1)(\bm{x}_1 - \bm{u}_1)\bm{x}_2\bm{z}_1(\bm{x}_2(\bm{u}_1-\bm{x}_1)\bm{x}_2\bm{z}_1 + \bm{y}_2\bm{z}_1 + 1)^{-1}]\bm{z}_1^{-1}.
	\end{equation}

 Finally, after substitution of \eqref{v1v2} into  \eqref{alg-var-e}, and replacement of $\bm{v_2}$ in the resulted equation by \eqref{v2}, it  follows that: 
 $$(\bm{u}_1-\bm{x}_1) (\bm{x}_2b\bm{z}_1 + \bm{x}_2\bm{y}_1\bm{y}_2\bm{z}_1 + \bm{y}_2\bm{z}_1 + \bm{x}_2\bm{y}_1 + 1 - a\bm{x}_2\bm{z}_1) = 0.$$
 The latter implies $\bm{u_1}=\bm{x}_1$. Then, from \eqref{v2}, \eqref{v1} and \eqref{w2}, we obtain $\bm{v}_2=\bm{y}_2$, $\bm{v}_1=\bm{y}_1$ and $\bm{w}_2=\bm{z}_2$.

We showed that system \eqref{alg-var} implies
$$\bm{u}_1=\bm{x}_1, \bm{u}_2=\bm{x}_2,\bm{v}_1=\bm{y}_1,\bm{v}_2=\bm{y}_2,\bm{w}_1=\bm{z}_1,\bm{w}_2=\bm{z}_2.$$
Thus, according to Proposition \ref{KP}, map \eqref{n_nls} is a  Yang--Baxter map.
\end{proof}

\subsection{Noncommutative Yang--Baxter map of DNLS type}
Let $\mathfrak{R}$ be a noncommutative division ring. We change $(\bm{f}\bm{p},\tilde{\bm{q}}\bm{f},c_1,c_2)\rightarrow (\bm{x}_1,\bm{x_2},1,1)$ in \eqref{DT-DNLS-1-pq}, namely we consider the following matrix
\begin{equation}\label{Lax-DNLS}
{\rm M} =\lambda^2\begin{pmatrix} a+\bm{x}_1\bm{x}_2 & 0\\ 0 & 0\end{pmatrix}+\lambda \begin{pmatrix} 0 & \bm{x}_1\\ \bm{x}_2 & 0 \end{pmatrix}+\begin{pmatrix} 1 & 0\\ 0 & 1\end{pmatrix}.
\end{equation}

We substitute ${\rm M}$ to the matrix refactorisation problem \eqref{eq-Lax} in order to generate a noncommutative Yang--Baxter map. In particular, we have the following. 

\begin{proposition} (Noncommutative DNLS map) Let $\mathfrak{A}=\mathfrak{R}\times\mathfrak{R}\times Z(\mathfrak{R})$. The map $Y_{a,b}:\mathfrak{A}^2 \rightarrow \mathfrak{A}^2$, given by
\begin{subequations}\label{n_dnls}
    \begin{align}
        \bm{x}_1\mapsto \bm{u}_1 &= \bm{y}_1+(a-b) \bm{x}_1 (a-\bm{y}_2 \bm{x}_1)^{-1},\label{n_dnls-a}\\
        \bm{x}_2\mapsto \bm{u}_2 &=\bm{y}_2\bm{x}_1\left[1-(a-b)(a-\bm{y}_2\bm{x}_1)^{-1}\right]^{-1}\bm{x}_1^{-1},\label{n_dnls-b}\\
        \bm{y}_1\mapsto \bm{v}_1 &=\bm{x}_1\left[1-(a-b)(a-\bm{y}_2\bm{x_1})^{-1}\right],\label{n_dnls-c}\\
        \bm{y}_2\mapsto \bm{v}_2 &=\bm{x}_2+\bm{y}_2-\bm{y}_2\bm{x}_1\left[1-(a-b)(a-\bm{y}_2\bm{x}_1)^{-1}\right]^{-1}\bm{x}_1^{-1},\label{n_dnls-d}
    \end{align}
\end{subequations}
has the following Lax representation
\begin{equation}\label{Lax-DNLS-YB}
    {\rm M}(\bm{u}_1,\bm{u}_2;a){\rm M}(\bm{v}_1,\bm{v}_2;b)={\rm M}(\bm{y}_1,\bm{y}_2;b){\rm M}(\bm{x}_1,\bm{x}_2;a),
\end{equation}
where $\rm{M}$ is given by \eqref{Lax-DNLS}, and admits the following functionally independent first integrals
\begin{equation}\label{1st_int-DNLS}
I_1=\bm{x}_1+\bm{y}_1 \quad\text{and}\quad I_2=\bm{x}_2+\bm{y}_2
\end{equation}
\end{proposition}
\begin{proof}
     The matrix refactorisation problem \eqref{Lax-Adler-Yamilov} is equivalent to the system of polynomial equations
    \begin{subequations}\label{DNLS-sys}
    \begin{align}
    (a+\bm{u}_1\bm{u}_2)(b+\bm{v}_1\bm{v}_2)&= (b+\bm{y}_1\bm{y}_2)(a+\bm{x}_1\bm{x}_2),\label{DNLS-sys-a}\\
\bm{u}_1\bm{u}_2+\bm{v}_1\bm{v}_2+\bm{u}_1\bm{v}_2&=\bm{y}_1\bm{y}_2+\bm{x}_1\bm{x}_2+\bm{y}_1\bm{x}_2,\label{DNLS-sys-b}\\
(a+\bm{u}_1\bm{u}_2)\bm{v}_1&=(b+\bm{y}_1\bm{y}_2)\bm{x}_1,\label{DNLS-sys-c}\\
\bm{u}_2(b+\bm{v}_1\bm{v}_2)&=\bm{y}_2(a+\bm{x}_1\bm{x}_2),\label{DNLS-sys-c-2}\\
\bm{u}_1+\bm{v}_1&=\bm{x}_1+\bm{y}_1,\label{DNLS-sys-d}\\
\bm{u}_2+\bm{v}_2&=\bm{x}_2+\bm{y}_2,\label{DNLS-sys-e}\\
\bm{u}_2\bm{v}_1&=\bm{y}_2\bm{x}_1.\label{DNLS-sys-f}
  \end{align}
    \end{subequations}

Equation \eqref{DNLS-sys-c} with use  of \eqref{DNLS-sys-f}  can be rewritten as  $a\bm{v}_1+\bm{u}_1\bm{y}_2\bm{x}_1=b\bm{x}_1+\bm{y}_1\bm{y}_2\bm{x}_1$. Substituting to  the latter equation $\bm{v}_1=\bm{x}_1+\bm{y}_1-\bm{u}_1$  from \eqref{DNLS-sys-d}, we  obtain
$$
\bm{u}_1(a-\bm{y}_2\bm{x_1})=(a-b)\bm{x}_1+\bm{y}_1(a-\bm{y}_2\bm{x_1}),
$$
which implies \eqref{n_dnls-a} after multiplication by $(a-\bm{y}_2\bm{x_1})^{-1}$ from the right. Then, from $\bm{v}_1=\bm{x}_1+\bm{y}_1-\bm{u}_1$, we obtain \eqref{n_dnls-c}. It  can be verified by straightforward computation that map \eqref{n_dnls} satisfies also equations \eqref{DNLS-sys-a}, \eqref{DNLS-sys-b} and \eqref{DNLS-sys-c-2}.

The invariants are obvious due to equations \eqref{DNLS-sys-d} and \eqref{DNLS-sys-e}.
\end{proof}

Next, we prove that map \eqref{n_dnls} satisfies the  Yang--Baxter equation. Specifically, we have the following.

\begin{theorem}
    Map \eqref{n_dnls} is a Yang--Baxter map.
\end{theorem}

\begin{proof}
    For matrix \eqref{DT-NLS} we consider the following  
     $$
    {\rm M}(\bm{u}_1,\bm{u}_2;a){\rm M}(\bm{v}_1,\bm{v}_2;b){\rm M}(\bm{z}_1,\bm{z}_2;c)={\rm M}(\bm{z}_1,\bm{x}_2;a){\rm M}(\bm{y}_1,\bm{y}_2;b){\rm M}(\bm{z}_1,\bm{z}_2;c).
    $$

    The above matrix trifactorisation problem is equivalent to the following system of polynomial equations:\allowdisplaybreaks
    \begin{subequations}\label{alg-var-dnls}
    \begin{align}
       & ac\bm{v}_1\bm{v}_2 + bc\bm{u}_1\bm{u}_2 + ab\bm{w}_1\bm{w}_2 + a\bm{v}_1\bm{v}_2\bm{w}_1\bm{w}_2 + b\bm{u}_1\bm{u}_2\bm{w}_1\bm{w}_2 + c\bm{u}_1\bm{u}_2\bm{v}_1\bm{v}_2 + \bm{u}_1\bm{u}_2\bm{v}_1\bm{v}_2\bm{w}_1\bm{w}_2 =\nonumber\\ &ac\bm{y}_1\bm{y}_2 + bc\bm{x}_1\bm{x}_2 + ab\bm{z}_1\bm{z}_2 + a\bm{y}_1\bm{y}_2\bm{z}_1\bm{z}_2 + b\bm{x}_1\bm{x}_2\bm{z}_1\bm{z}_2 + c\bm{x}_1\bm{x}_2\bm{y}_1\bm{y}_2 + \bm{x}_1\bm{x}_2\bm{y}_1\bm{y}_2\bm{z}_1\bm{z}_2,\label{alg-var-dnls-a}\\
       &a\bm{v}_1\bm{v}_2 + b\bm{u}_1\bm{u}_2 + c\bm{v}_1\bm{v}_2 + c\bm{u}_1\bm{u}_2 + \bm{u}_1\bm{u}_2\bm{v}_1\bm{v}_2 + c\bm{u}_1\bm{v}_2 + b\bm{w}_1\bm{w}_2 + \bm{v}_1\bm{v}_2\bm{w}_1\bm{w}_2 + \bm{u}_1\bm{u}_2\bm{w}_1\bm{w}_2 +\nonumber\\ 
       &\bm{u}_1\bm{v}_2\bm{w}_1\bm{w}_2 + a\bm{v}_1\bm{w}_2 + \bm{u}_1\bm{u}_2\bm{v}_1\bm{w}_2 + a\bm{w}_1\bm{w}_2 = a\bm{y}_1\bm{y}_2 + b\bm{x}_1\bm{x}_2 + c\bm{y}_1\bm{y}_2 + c\bm{x}_1\bm{x}_2 + \bm{x}_1\bm{x}_2\bm{y}_1\bm{y}_2 + \nonumber\\
       &c\bm{x}_1\bm{y}_2 + b\bm{z}_1\bm{z}_2 + \bm{y}_1\bm{y}_2\bm{z}_1\bm{z}_2 + \bm{x}_1\bm{x}_2\bm{z}_1\bm{z}_2 + \bm{x}_1\bm{y}_2\bm{z}_1\bm{z}_2 + a\bm{y}_1\bm{z}_2 + \bm{x}_1\bm{x}_2\bm{y}_1\bm{z}_2 + a\bm{z}_1\bm{z}_2,\label{alg-var-dnls-b}\\
       &\bm{v}_1\bm{v}_2 + \bm{u}_1\bm{u}_2 + \bm{u}_1\bm{v}_2 + \bm{w}_1\bm{w}_2 + \bm{v}_1\bm{w}_2 + \bm{u}_1\bm{w}_2 = \bm{y}_1\bm{y}_2 + \bm{x}_1\bm{x}_2 + \bm{x}_1\bm{y}_2 + \bm{z}_1\bm{z}_2 + \bm{y}_1\bm{z}_2 + \bm{x}_1\bm{z}_2,\label{alg-var-dnls-c}\\
       &ab\bm{w}_1 + a\bm{v}_1\bm{v}_2\bm{w}_1 + b\bm{u}_1\bm{u}_2\bm{w}_1 + \bm{u}_1\bm{u}_2\bm{v}_1\bm{v}_2\bm{w}_1 = ab\bm{z}_1 + a\bm{y}_1\bm{y}_2\bm{z}_1 + b\bm{x}_1\bm{x}_2\bm{z}_1 + \bm{x}_1\bm{x}_2\bm{y}_1\bm{y}_2\bm{z}_1,\label{alg-var-dnls-d}\\
       &b\bm{w}_1 + \bm{v}_1\bm{v}_2\bm{w}_1 + \bm{u}_1\bm{v}_2\bm{w}_1 + a\bm{v}_1 + \bm{u}_1\bm{u}_2\bm{v}_1 + a\bm{w}_1 + \bm{u}_1\bm{u}_2\bm{w}_1 =\nonumber\\
       & b\bm{z}_1 + \bm{y}_1\bm{y}_2\bm{z}_1 + \bm{x}_1\bm{y}_2\bm{z}_1 + a\bm{y}_1 + \bm{x}_1\bm{x}_2\bm{y}_1 + a\bm{z}_1 + \bm{x}_1\bm{x}_2\bm{z}_1,\label{alg-var-dnls-e}\\
       &\bm{w}_1 + \bm{v}_1 + \bm{u}_1 = \bm{z}_1 + \bm{y}_1 + \bm{x}_1,\label{alg-var-dnls-f}\\
       &cb\bm{u}_2 + b\bm{u}_2\bm{w}_1\bm{w}_2 + c\bm{u}_2\bm{v}_1\bm{v}_2 + \bm{u}_2\bm{v}_1\bm{v}_2\bm{w}_1\bm{w}_2 = cb\bm{x}_2 + b\bm{x}_2\bm{z}_1\bm{z}_2 + c\bm{x}_2\bm{y}_1\bm{y}_2 + \bm{x}_2\bm{y}_1\bm{y}_2\bm{z}_1\bm{z}_2,\label{alg-var-dnls-g}\\
       &c\bm{u}_2 + \bm{u}_2\bm{w}_1\bm{w}_2 + \bm{u}_2b + \bm{u}_2\bm{v}_1\bm{v}_2 + \bm{u}_2\bm{v}_1\bm{w}_2 + c\bm{v}_2 + \bm{v}_2\bm{w}_1\bm{w}_2 = \nonumber\\
       &c\bm{x}_2 + \bm{x}_2\bm{z}_1\bm{z}_2 + b\bm{x}_2 + \bm{x}_2\bm{y}_1\bm{y}_2 + \bm{x}_2\bm{y}_1\bm{z}_2 + \bm{y}_2\bm{z}_1\bm{z}_2 + c\bm{y}_2,\label{alg-var-dnls-h}\\
       &\bm{u}_2 + \bm{w}_2 + \bm{v}_2 = \bm{x}_2 + \bm{z}_2 + \bm{y}_2,\label{alg-var-dnls-i}\\
       &b\bm{u}_2\bm{w}_1 + \bm{u}_2\bm{v}_1\bm{v}_2\bm{w}_1 = b\bm{x}_2\bm{z}_1 + \bm{x}_2\bm{y}_1\bm{y}_2\bm{z}_1,\label{alg-var-dnls-j}\\
       &\bm{u}_2\bm{w}_1 + \bm{v}_2\bm{w}_1 + \bm{u}_2\bm{v}_1 = \bm{x}_2\bm{z}_1 + \bm{y}_2\bm{z}_1 + \bm{x}_2\bm{y}_1.\label{alg-var-dnls-k}
    \end{align}  
    \end{subequations}

    For simplicity, we introduce the notation $\bm{u}_2\bm{v}_1 \rightarrow Q$,\quad $\bm{u}_2\bm{w}_1 \rightarrow S$,\quad $\bm{v}_2\bm{w}_1 \rightarrow F$, and we rewrite  equations \eqref{alg-var-dnls-k}, \eqref{alg-var-dnls-j}, \eqref{alg-var-dnls-h} and \eqref{alg-var-dnls-g}, namely:
    \begin{align}
       & F + Q + S - \bm{x}_2\bm{y}_1 - \bm{x}_2\bm{z}_1 - \bm{y}_2\bm{z}_1 = 0 ,\label{eq-1}\\
      &  QF + bS - b\bm{x}_2\bm{z}_1 - \bm{x}_2\bm{y}_1\bm{y}_2\bm{z}_1 = 0,\label{eq-2}\\
      &  b\bm{u}_2 + c(\bm{u}_2 + \bm{v}_2) + Q\bm{v}_2 + (F+Q+S)\bm{w}_2 - b\bm{x}_2 - c\bm{x}_2 - c\bm{y}_2 - \bm{x}_2\bm{y}_1\bm{y}_2 - \bm{x}_2\bm{y}_1\bm{z}_2 - \bm{x}_2\bm{z}_1\bm{z}_2 - \bm{y}_2\bm{z}_1\bm{z}_2 = 0,\label{eq-3}\\
      &  bc\bm{u}_2 + cQ\bm{v}_2 + (QF + bS)\bm{w}_2 - bc\bm{x}_2 - c\bm{x}_2\bm{y}_1\bm{y}_2 - b\bm{x}_2\bm{z}_1\bm{z}_2 - \bm{x}_2\bm{y}_1\bm{y}_2\bm{z}_1\bm{z}_2 = 0.\label{eq-4}
    \end{align}

    Solving \eqref{alg-var-dnls-i} for $\bm{u}_2+\bm{v}_2$, \eqref{eq-1} for $F + Q + S$ and substituting into \eqref{eq-3}, we obtain:
    \begin{equation}\label{4_dnls}
	b\bm{u}_2 + Q\bm{v}_2 + c(-\bm{w}_2 + \bm{z}_2) + (\bm{x}_2\bm{y}_1 + \bm{x}_2\bm{z}_1 + \bm{y}_2\bm{z}_1)\bm{w}_2 - b\bm{x}_2 - \bm{x}_2\bm{y}_1\bm{y}_2 - \bm{x}_2\bm{y}_1\bm{z}_2 - \bm{x}_2\bm{z}_1\bm{z}_2 - \bm{y}_2\bm{z}_1\bm{z}_2 = 0.
	\end{equation}

 From \eqref{4_dnls} and \eqref{alg-var-dnls-i} we obtain:
 \begin{equation}\label{w2_dnls}
	\bm{w}_2 = \bm{z}_2 + (-c + \bm{x}_2\bm{y}_1 + \bm{x}_2\bm{z}_1 + \bm{y}_2\bm{z}_1)^{-1}(-b\bm{u}_2 + Q(-\bm{x}_2 - \bm{y}_2 - \bm{z}_2 + \bm{u}_2 + \bm{w}_2) + b \bm{x}_2 + \bm{x}_2\bm{y}_1\bm{y}_2).
	\end{equation}

 Substituting to \eqref{eq-4} $QF+bS=b\bm{x}_2\bm{z}_1 + \bm{x}_2\bm{y}_1\bm{y}_2\bm{z}_1$ from \eqref{eq-2}, and  then $\bm{v}_2$ from \eqref{alg-var-dnls-j}, we obtain:
\begin{equation}\label{u2_dnls}
\bm{u}_2 = (cb-cQ)^{-1} (cQ(\bm{w}_2-\bm{x}_2-\bm{y}_2-\bm{z}_2) - b\bm{x}_2\bm{z}_1\bm{w}_2 - \bm{x}_2\bm{y}_1\bm{y}_2\bm{z}_1\bm{w}_2  +bc\bm{x}_2 + c\bm{x}_2\bm{y}_1\bm{y}_2 + b\bm{x}_2\bm{z}_1\bm{z}_2 + \bm{x}_2\bm{y}_1\bm{y}_2\bm{z}_1\bm{z}_2).
\end{equation}
Now, substitution of  \eqref{u2_dnls} to \eqref{w2_dnls} implies:
$$\left[1+(-c + \bm{x}_2\bm{y}_1 + \bm{x}_2\bm{z}_1 + \bm{y}_2\bm{z}_1)^{-1}\left(-\frac{b}{c} \bm{x}_2\bm{z}_1 - \frac{1}{c}\bm{x}_2\bm{y}_1\bm{y}_2\bm{z}_1\right)\right](\bm{z}_2-\bm{w}_2) = 0, $$
from where it follows that $\bm{w}_2=\bm{z}_2$.

Next,  from \eqref{4_dnls}, given that $\bm{w}_2=\bm{z}_2$, we obtain $\bm{u}_2\bm{v}_1\bm{v}_2 = -b\bm{u}_2 + b\bm{x}_2 + \bm{x}_2\bm{y}_1\bm{y}_2$. We substitute the latter into  \eqref{alg-var-dnls-j} to  obtain:
$$
(b\bm{x}_2 + \bm{x}_2\bm{y}_1\bm{y}_2)(\bm{w}_1-\bm{z}_1)=0,
$$
from where it follows that $\bm{w}_1=\bm{z}_1$.

Now, we substitute $\bm{w}_2=\bm{z}_2$ to \eqref{alg-var-dnls-i} to  obtain: $\bm{v}_2+\bm{u}_2=\bm{y}_2+\bm{x}_2$. Substitution of the latter together  with $\bm{w}_1=\bm{z}_1$ to \eqref{alg-var-dnls-k} implies: $\bm{u}_2\bm{v}_1=\bm{x}_2\bm{y}_1$. Then,  from \eqref{u2_dnls} follows
$\bm{u}_2 = (b-\bm{x}_2\bm{y}_1)^{-1}(b-\bm{x}_2\bm{y}_1)\bm{x}_2=\bm{x}_2$.

Finally, from \eqref{alg-var-dnls-i}, \eqref{alg-var-dnls-k}  and \eqref{alg-var-dnls-f}, we obtain $\bm{v}_2 = \bm{y}_2$, $\bm{v}_1=\bm{y}_1$ and $\bm{u}_1=\bm{x}_1$, respectively.

That  is, we showed that system \eqref{alg-var-dnls} implies
$$\bm{u}_1=\bm{x}_1, \bm{u}_2=\bm{x}_2,\bm{v}_1=\bm{y}_1,\bm{v}_2=\bm{y}_2,\bm{w}_1=\bm{z}_1,\bm{w}_2=\bm{z}_2.$$
Thus, according to Proposition \ref{KP}, map \eqref{n_dnls} is a  Yang--Baxter map.
\end{proof}

\section{Conclusions}\label{conclusions}
In this paper, we derived new solutions to the famous Yang--Baxter equation, and also constructed noncommutative versions of well-known Yang--Baxter  maps. In particular, using the approach of correspondence \cite{Igonin-Sokor}, we constructed new Yang--Baxter maps of  KdV  type, namely maps \eqref{KdV-1}  and \eqref{KdV-2}. We showed that map  \eqref{KdV-1} is completely integrable in the Liouville sense. Then, we constructed a fully noncommutative version of the ``KdV lift'' \cite{Kouloukas} which we squeezed down to the noncommutative discrete potential KdV equation.

Moreover, we constructed a noncommutative versions of certain Darboux transformations for the derivative NLS  equation \cite{SPS}, namely transformations \eqref{DT-DNLS-1}--\eqref{BT-DNLS-1} and \eqref{DT-DNLS-2}--\eqref{BT-DNLS-2}. Then,  using a first  integral of the  system \eqref{BT-DNLS-1}, we derived a parametric  family of Darboux matrices for the  noncommutative  derivative NLS  equation, namely Darboux matrices \eqref{DT-DNLS-1-pq}.

Finally, we employed a Darboux matrix for the  noncommutative NLS equation \cite{SP} and matrix \eqref{DT-DNLS-1-pq}  to contruct noncommutative maps (maps with  their  variables belonging to a division ring) of NLS and  derivative NLS type, namely maps \eqref{n_nls} and \eqref{n_dnls}, which  we showed  that are Yang--Baxter maps.

Our results can be extended in the following  ways:
\begin{itemize}
    \item Study  the Liouville integrability of maps \eqref{KdV-YB-NC}, \eqref{n_nls} and \eqref{n_dnls}. The amount of invariants (first integrals) of map \eqref{n_dnls} is enough  for complete  integrability  claims. However, for maps \eqref{KdV-YB-NC} and \eqref{n_nls}, one must  first find  one more  first integral. 

    \item Matrices \eqref{Lax-NLS} and \eqref{Lax-DNLS} come from more general Darboux transformations after using first integrals. Similarly to the  commutative case \cite{Sokor-Sasha}, one  could employ the more general  Darboux matrices in order to derive noncommutative six-dimensional  Yang--Baxter maps  which can be restricted on  certain invariant leaves to the noncommutative four-dimensional Yang--Baxter maps \eqref{n_nls} and \eqref{n_dnls}.

    \item The elements of Yang--Baxter maps \eqref{n_nls} and \eqref{n_dnls} are solutions to the NLS and the derivative DNLS  equation.  Maps \eqref{n_nls} and \eqref{n_dnls} may  be  certain B\"acklund  transformations for the  NLS and the derivative DNLS  equation, respectively.

   \item Study the linearised version of map \eqref{n_dnls}. The linearisation procedure was shown for the derivative NLS map in the  commutative case  in \cite{BIKRP}.

    \item Employ Darboux matrices of  NLS type to construct noncommutative (on  division rings) versions of the NLS entwining Yang--Baxter maps  which appeared in  \cite{Sokor-Pap}. Note that Grassmann extended versions of these maps already appeard  in \cite{Giota-Miky}  and are the first  examples of Grassmann  extended entwining Yang--Baxter maps in the literature.

    \item The Yang--Baxter equation is the second member of a  wider family of equations, the $n$-simplex equations. Employ the matrices considered in this paper to construct noncommutative solutions for higher members in the family of $n$-simplex equations. A noncommutative NLS  type  tetrahedron map already  appeared in  \cite{Sokor-2022}.

    \item In this paper we used the method presented in \cite{Kouloukas} in order to associate the our noncommutative KdV type of Yang--Baxter map to the  noncommutative discrete potential KdV equation. One could use the  symmetry approach \cite{Pap-Tongas-Veselov, Kassotakis-Tetrahedron}, or the method of integrals in separated variables \cite{Pavlos-Maciej-2} to relate the  KdV maps to other noncommutative integrable lattice equations.
\end{itemize}

\section*{Acknowledgements}
This work was funded by the Russian Science Foundation project No. 20-71-10110 (https://rscf.ru/en/project/23-71-50012/).
Part of this  work,  namely the proofs of Theorems 3.3, 4.4 and 4.6, was carried out within the framework of a development programme for the Regional Scientific and Educational Mathematical Centre of the P.G. Demidov Yaroslavl State University with financial support from the Ministry of Science and Higher Education of the Russian Federation (Agreement on provision of subsidy from the federal budget No. 075-02-2023-948).

\end{document}